\documentclass[letterpaper,journal]{IEEEtran}
\usepackage{amsmath,amsfonts,amssymb}
\usepackage{algorithmic}
\usepackage{algorithm}
\usepackage{array}
\usepackage[caption=false,font=normalsize,labelfont=sf,textfont=sf]{subfig}
\usepackage{textcomp}
\usepackage{stfloats}
\usepackage{url}
\usepackage{verbatim}
\usepackage{graphicx}
\usepackage{pifont}
\usepackage{color}
\usepackage{upgreek}
\usepackage{cite}
\newtheorem{theorem}{Theorem}

\newtheorem{assumption}{Assumption}
\newtheorem{proposition}{Proposition}
\newtheorem{corollary}{Corollary}
\newtheorem{remark}{Remark}
\def\BibTeX{{\rm B\kern-.05em{\sc i\kern-.025em b}\kern-.08em
 T\kern-.1667em\lower.7ex\hbox{E}\kern-.125emX}}
\usepackage{balance}
\begin{document}
\title{When Do Pilots Matter for OFDM Sensing? Pilot-Data Resource Design for ISAC}

\author{\IEEEauthorblockN{
 Shengcai Zhou, \emph{Student~Member, IEEE},
 Luping Xiang, \emph{Senior Member, IEEE},
 Yi Wang, \emph{Member, IEEE},\\
 and Kun Yang, \emph{Fellow, IEEE}}
 \vspace{-0.5 cm}\\
 \thanks{S. Zhou, L. Xiang and K. Yang are with the State Key Laboratory of Novel Software Technology, Nanjing University, Nanjing 210008, China, and School of Intelligent Software and Engineering, Nanjing University (Suzhou Campus), Suzhou 215163, China, email: zhshc@smail.nju.edu.cn, luping.xiang@nju.edu.cn, kunyang@nju.edu.cn.}
 \thanks{Y. Wang is with the Yangtze Delta Region Institute (Quzhou), University of Electronic Science and Technology of China, Quzhou 324003, China (e-mail: wangyi@csj.uestc.edu.cn).}
}


\maketitle

\begin{abstract}
Practical OFDM-based integrated sensing and communication (ISAC) signals contain deterministic pilots and random data payloads, yet how these two types of resources jointly affect matched-filter sensing performance remains insufficiently understood. This paper establishes an analytical and optimization framework for pilot-data (P-D) OFDM sensing. We first derive a closed-form mean-square periodic autocorrelation function (MS P-ACF), which characterizes the respective contributions of the average subcarrier power spectrum and data-symbol randomness to sensing sidelobes. Based on this result, we reveal that pilot placement has no impact on normalized sidelobes under uniform power allocation or on the normalized global expected integrated sidelobe level (EISL), whereas it becomes an effective design degree of freedom for reshaping sidelobes under nonuniform power allocation. We further derive the globally optimal two-level pilot-data power allocation for global EISL minimization and characterize pilot placement for peak sidelobe suppression. For region-of-interest (ROI) sensing, we jointly optimize pilot placement and subcarrier power allocation to minimize the normalized ROI-EISL. A strictly uniform pilot placement is proved globally optimal under specified conditions, while a bisection-and-branch-and-bound algorithm is developed for general resource configurations. Numerical results validate the theoretical analysis and demonstrate that the proposed design reduces the ROI sidelobe level by up to 2.19 dB compared with optimized uniform-pilot OFDM, while improving ranging accuracy at an SNR of -10 dB by 84.3$\%$. These results reveal when and how pilot placement can be exploited to regulate the sensing behavior of random OFDM communication signals, providing practical resource-design principles for OFDM-based ISAC.
\end{abstract}

\begin{IEEEkeywords}
ISAC, OFDM, pilot-data signaling, periodic autocorrelation function, resource optimization
\end{IEEEkeywords}

\section{Introduction}
\subsection{Background}
\IEEEPARstart{A}{s} integrated sensing and communication (ISAC) research advances from concept validation to prototype development, sensing with communication signals is becoming an increasingly practical approach to ISAC implementation \cite{Wei2023ISACSignals,Zhou2026ExtendedTarget,Luo2026Bedrock}. Existing cellular networks not only provide wide-area coverage for sensing \cite{Xie2022Collaborative}, but also enable the reuse of communication signals while maintaining compatibility with existing standards; by contrast, redesigning ISAC systems entails challenges such as interference management and equipment modification \cite{Lu2025SensingRandom}. Unlike conventional deterministic sensing waveforms, the randomness of data carried by communication signals causes the actual transmitted waveform to vary with the transmitted content\cite{luo2026dbuofdm}. When random data payloads are directly exploited for sensing, the autocorrelation function and sensing metrics for detection and estimation become data-dependent, thereby degrading the reliability of target detection and parameter estimation. Although randomness is difficult to eliminate, communication systems can still regulate the sensing performance of random signals through physical layer design and receiver processing degrees of freedom, including constellation and modulation design, pulse shaping, power allocation, spatial precoding, and sensing receiver processing.

Recent studies have investigated communication-centric sensing from different perspectives. Zhang \emph{et al.} revealed the role of subcarrier power allocation in controlling ranging sidelobes \cite{Zhang2026OptimalPower}. Building on this, Meng \emph{et al.} jointly optimized constellation selection and subcarrier power allocation \cite{Meng2026ConstellationSelection}. Pilots also constitute an important resource for OFDM sensing. Bouziane and Arslan analyzed the ambiguity function (AF) grating lobes induced by uniform pilot placement \cite{Bouziane2026OptimizedPilot}. Furthermore, Xu \emph{et al.} exploited deterministic pilots and random data payloads for multiantenna target response estimation \cite{Xu2025PilotsDataPayloads}. These studies show that both random data and deterministic pilots can serve as sensing resources in practical OFDM frames. However, how they jointly affect matched-filter sensing performance and can be further optimized remains unclear.

\subsection{Related Work}
Systematic investigations of communication signal sensing began with the OFDM radar communication (RadCom) framework, which employed matched filtering and subcarrier-wise division filtering with two-dimensional Fourier processing for range and relative velocity estimation. Prototype experiments validated the feasibility of OFDM sensing while also revealing high sidelobes in its time domain autocorrelation \cite{sturm2009ofdm,sturm2011waveform}.

A rigorous theoretical characterization of sensing losses caused by communication randomness has long been lacking. Xiong \emph{et al.} revealed the tradeoff between deterministic and random signaling between communication and sensing by constructing a Cram\'er-Rao bound (CRB)-rate region: communication relies on input randomness to convey information, whereas sensing favors signal realizations with deterministic sample covariance \cite{Xiong2023FundamentalTradeoff}. This result provides a theoretical foundation for quantifying sensing losses induced by random signals.

Building on this, Xie \emph{et al.} introduced sensing mutual information (SMI) and established its upper bound for deterministic signals, although SMI is difficult to apply directly in practical sensing \cite{Xie2025SensingMutualInformation}. For monostatic multiple-input multiple-output (MIMO) sensing, the ergodic linear minimum mean-square error (ELMMSE) metric shows that precoding based on deterministic covariance underestimates target response estimation errors for short frames, but requires computationally expensive matrix inversion \cite{Lu2024RandomPrecoding}. For low-complexity matched filtering, the variance of the AF quantifies losses induced by randomness. The controllable component of this variance is determined by the constellation symbol power variance $\mathbb{E}[|x|^4]-1$. A larger fluctuation leads to higher average sidelobe power and poorer weak target detection performance \cite{Du2024Reshaping}. Furthermore, Liu \emph{et al.} employed the periodic autocorrelation function (P-ACF) and aperiodic autocorrelation function (A-ACF) to characterize the cyclic and linear matched filter responses of waveforms with and without a cyclic prefix (CP), respectively. Under quadrature amplitude modulation (QAM) and phase-shift keying (PSK), they proved that CP-OFDM is the unique globally optimal waveform that minimizes the average ranging sidelobes among all orthogonal communication waveforms with a CP \cite{Liu2025CPOFDM}. 

For waveforms with a CP, Liu \emph{et al.} proved that, under QAM/PSK constellations and Nyquist shaping, OFDM minimizes the expected P-ACF sidelobe level at every lag and can be further optimized through Nyquist pulse shaping \cite{Liu2025Iceberg}. Subcarrier power allocation is another important design degree of freedom: under full bandwidth transmission, uniform power allocation (UPA) globally minimizes the expected integrated sidelobe level (EISL) and the average P-ACF sidelobe at every nonzero lag \cite{Zhang2026OptimalPower}. In addition, prior information such as geographical information and historical scan results can constrain targets to local regions \cite{Wang2022LocalAmbiguity} and enable the suppression of data-dependent sidelobes over specified delay intervals using region of interest (ROI) mismatched filters \cite{Yang2026ConstellationIndependent}. However, existing studies mainly consider waveforms with purely random data and have neither characterized the P-ACF of practical pilot-data (P-D) OFDM signals nor optimized their sensing performance within the ROI.

On the other hand, sensing with deterministic signals in practical communication frames has also been extensively studied \cite{Li2025RethinkingISACSignaling}. Kumari \emph{et al.} reused the Golay complementary sequences in the IEEE 802.11ad frame preamble for radar probing, providing an early demonstration that deterministic training sequences in communication frames can directly serve as matched filter radar waveforms \cite{Kumari2018IEEE80211adRadar}. Subsequent studies have extended this approach to various communication reference signals, including the positioning reference signal (PRS), sounding reference signal (SRS), demodulation reference signal (DMRS), and channel state information reference signal (CSI-RS) \cite{Wei2022PRSBasedSensing,Wei2023DMRSBasedISAC}. However, deterministic signals usually occupy only a limited fraction of the available time-frequency resources, resulting in limited sensing performance due to insufficient effective bandwidth and sparse temporal observations.

Related pilot placement principles have also been studied in OFDM timing synchronization. The periodic and aperiodic autocorrelation sidelobe energies of pilot subcarrier waveforms under one-dimensional matched filtering were derived, establishing the fundamental relationship between deterministic pilot placement and the autocorrelation sidelobes of synchronization waveforms \cite{Ureten2009Autocorrelation}. Subsequently, Zhang \emph{et al.} analyzed the effects of different patterns of unified time and frequency reference signals, namely generalized pilots, on two-dimensional delay-Doppler ambiguity performance and resource overhead, and provided pattern design criteria for different sensing algorithms \cite{Zhang2024OFDMReferenceSignalPattern}. Li \emph{et al.} also proposed a similar framework \cite{Li2025SensingOrientedAdaptive}. Bouziane and Arslan further optimized the range grating lobes caused by uniform pilot placement, which had been identified previously, subject to communication constraints \cite{Bouziane2026OptimizedPilot}. Overall, sensing schemes relying solely on deterministic signals cannot overcome the sensing performance ceiling imposed by their limited occupancy.

Recent studies have jointly considered both signal types. From an information-theoretic perspective, Xie \emph{et al.} investigated sensing performance with deterministic pilots superposed on random data using SMI and developed an optimization approach for SMI maximization in bistatic multiantenna ISAC systems \cite{Xie2026SensingMutualInformation}. Xu \emph{et al.} further analyzed estimation performance using the ELMMSE metric and optimized the spatial precoding of data signals \cite{Xu2025PilotsDataPayloads}. Du \emph{et al.} considered unified P-D design for practical orthogonal time frequency space (OTFS) waveforms, adopted an EISL metric that accounts for P-D resources, and jointly designed deterministic pilot symbols and data power allocation \cite{Du2025UniformPilotDataOTFS}. Nevertheless, matched filter sensing performance for P-D OFDM signals has not been systematically analyzed or optimized.

\subsection{Contributions}
\begin{table*}[t]
\centering
\caption{Contrasting Our Contributions To The State of the Art}
\setlength{\tabcolsep}{4mm}{
\begin{tabular}{l|c|c|c|c|c|c}
\hline 
Contributions & \textbf{This work} & \cite{Zhang2026OptimalPower} & \cite{Bouziane2026OptimizedPilot}
 &\cite{Meng2026ConstellationSelection}
 & \cite{Du2025UniformPilotDataOTFS}
 & \cite{Yang2026ConstellationIndependent}
\\ \hline\hline
OFDM & \ding{52} & \ding{51} & \ding{51} & \ding{51} & & \ding{51} \\
\hline
Matched filter sensing & \ding{52} & \ding{51} & \ding{51} & \ding{51} & \ding{51} & \\
\hline
ROI & \ding{52} & & & & & \ding{51} \\
\hline
Pilot optimization & \ding{52} & & \ding{51} & & \ding{51} & \\
\hline
Subcarrier power allocation & \ding{52} & \ding{51} & & \ding{51} & & \\
\hline

\end{tabular}
}
\vspace{-0.3 cm}
\label{contributions}
\end{table*}
Motivated by the above studies, we establish a sensing performance analysis framework for practical P-D OFDM signals and propose a joint resource optimization strategy for minimizing the normalized form of the region of interest expected integrated sidelobe level (ROI-EISL). 
The novelty of this work is summarized in Table~\ref{contributions}, and the main contributions are as follows:

\begin{itemize}
\item We establish an analytical framework for the matched filter ranging performance of P-D OFDM signals. A closed-form expression for the mean-square P-ACF (MS P-ACF) is derived, which separates the deterministic contribution of the average P-D power spectrum from the contribution induced by data randomness, thereby enabling subsequent sidelobe analysis and optimization.

\item We characterize how the P-D resource structure governs normalized MS P-ACF sidelobes. We prove pilot placement invariance under uniform power allocation and derive the globally optimal two-level power allocation for normalized global EISL minimization. We further show that, under nonuniform power allocation, pilot placement reshapes the sidelobe distribution and can be exploited to reduce the peak sidelobe level ratio (PSLR).

\item We develop an ROI-aware P-D resource optimization framework for normalized ROI-EISL minimization. Under a uniform data subcarrier power lower bound and the condition $\xi<|\mathcal{S}_{\mathrm{p}}|$, we establish a globally optimal strictly uniform pilot placement. For general ROIs and power constraints, we formulate the joint design as a mixed-integer quadratic fractional program and develop a bisection-and-branch-and-bound (BBnB) algorithm for its solution.

\item Numerical results verify the derived effects of the P-D resource structure and the ROI-aware optimization results, and demonstrate the resulting gains in sidelobe suppression and ranging performance under different ROIs, pilot configurations, power constraints, and system parameters.
\end{itemize}

The remainder of this paper is organized as follows. Section~II presents the system model. Section~III analyzes the effects of the P-D resource structure on normalized MS P-ACF sidelobes. Section IV presents the ROI-aware P-D resource optimization problem and its solution. Section~V presents numerical results. Section~VI concludes the paper and discusses future research directions.

\textit{Notation:} Throughout this paper, lowercase and uppercase boldface
letters denote vectors and matrices, respectively.
$\mathbb{E}[\cdot]$ denotes statistical expectation,
$\lvert z\rvert$ denotes the absolute value or magnitude of a scalar $z$,
and $\lvert\mathcal{A}\rvert$ denotes the cardinality of a set $\mathcal{A}$.
Moreover, $(\cdot)^*$, $(\cdot)^{\mathrm{T}}$, and $(\cdot)^{\mathrm{H}}$ denote complex
conjugation, transpose, and conjugate transpose, respectively,
while $[\boldsymbol{A}]_{m,n}$ denotes the $(m,n)$th entry of matrix $\boldsymbol{A}$.
$\mathbf{1}$ and $\mathbf{0}$ denote the all-one and all-zero vectors of
appropriate dimensions, respectively.
For a Hermitian matrix $\boldsymbol{A}$, $\boldsymbol{A}\succeq\mathbf{0}$ denotes
that $\boldsymbol{A}$ is positive semidefinite.
Finally, $\mathrm{j}=\sqrt{-1}$ denotes the imaginary unit, and
$a\equiv b\pmod K$ denotes congruence modulo $K$.

\section{System Model}
We consider a single-antenna monostatic ISAC system that reuses the OFDM signal transmitted to a communication receiver for matched filter ranging within the ROI. The communication transmitter and sensing receiver (SRx) are colocated. The transmitted signal is received by the communication receiver and reflected by potential targets, with the resulting echoes collected and processed by the SRx.
\subsection{P-D OFDM Signal Model}
Consider an OFDM symbol with $K$ subcarriers. Full bandwidth transmission is assumed throughout, such that all subcarriers are activated and partitioned into a pilot set $\mathcal{S}_{\mathrm{p}}$ and a data set $\mathcal{S}_{\mathrm{d}}$, where $\mathcal{S}_{\mathrm{p}}\cup\mathcal{S}_{\mathrm{d}}=\mathcal{S}=\{0,1,\dots,K-1\}$ and $\mathcal{S}_{\mathrm{p}}\cap\mathcal{S}_{\mathrm{d}}=\varnothing$. Unless otherwise specified, the cardinalities of $\mathcal{S}_{\mathrm{p}}$ and $\mathcal{S}_{\mathrm{d}}$ are determined by a given resource configuration. The general form of the modulated symbol on the $k$th subcarrier is given by
\begin{equation}
X[k] =
\begin{cases}
\sqrt{p_{\mathrm{p},k}}a_k, & k\in\mathcal{S}_{\mathrm{p}},\\
\sqrt{p_{\mathrm{d},k}}s_k, & k\in\mathcal{S}_{\mathrm{d}},
\end{cases},
\end{equation}
where $p_{\mathrm{p},k}$ and $p_{\mathrm{d},k}$ denote the powers allocated to the $k$th pilot and data subcarriers, respectively; $a_k$ denotes the $k$th pilot symbol and satisfies $|a_k|^2=1$; and $s_k$ denotes the $k$th data symbol. The total transmit power is $P=\sum_{k\in\mathcal{S}_{\mathrm{p}}}p_{\mathrm{p},k}+\sum_{k\in\mathcal{S}_{\mathrm{d}}}p_{\mathrm{d},k}$.

\begin{assumption}[Unit Power and Rotational Symmetry]\label{assume1}
The data symbols are mutually independent and satisfy the unit power, zero mean, and zero pseudovariance conditions
\begin{equation}
\mathbb{E}\!\left[|s_k|^2\right] = 1,~
\mathbb{E}[s_k] = 0,~
\mathbb{E}\!\left[s_k^2\right] = 0,~
\forall k\in\mathcal{S}_{\mathrm{d}}.
\end{equation}
\end{assumption}

Assumption~\ref{assume1} holds for most practical constellations, including the PSK and QAM families, except for binary phase-shift keying and 8-QAM. Next, we characterize random fluctuations in symbol power using the fourth-order moment $\mu_4\triangleq \mathbb{E}[|s_k|^4]$ of the data symbols. For constant-modulus constellations, $\mu_4=1$, and for the considered nonconstant-modulus QAM constellations, $1<\mu_4<2$.

\subsection{Communication Receiver Model}
Assuming that the CP of the transmitted signal is no shorter than the maximum channel delay spread, after CP removal and $K$-point discrete Fourier transform (DFT) processing, the received communication signal on the $k$th subcarrier is expressed as
\begin{equation}
Y_{\mathrm{c}}[k]
=
H_{\mathrm{c}}[k]X[k]+N_{\mathrm{c}}[k],
\end{equation}
where $H_{\mathrm{c}}[k]$ denotes the frequency domain communication channel coefficient on the $k$th subcarrier, and $N_{\mathrm{c}}[k]\sim\mathcal{CN}(0,\sigma_{\mathrm{c}}^2)$ denotes frequency domain additive white Gaussian noise with variance $\sigma_{\mathrm{c}}^2$. Communication information is carried only on data subcarriers, and the signal-to-noise ratio (SNR) on the $k$th data subcarrier is $\gamma_k=
|H_{\mathrm{c}}[k]|^2p_{\mathrm{d},k}/\sigma_{\mathrm{c}}^2$. Following \cite{Zhou2014WirelessInformation}, we use the following communication rate metric, obtained by summing the subcarrier rates over all data subcarriers:
\begin{equation}
R_{\mathrm{c}}
=
\sum_{k\in\mathcal{S}_{\mathrm{d}}}
\log_2\left(
1+
\frac{|H_{\mathrm{c}}[k]|^2p_{\mathrm{d},k}}
{\sigma_{\mathrm{c}}^2}
\right).
\end{equation}
\subsection{Sensing Receiver Model}
For simplicity, we consider a stationary single target with normalized round-trip delay index $b_0$. After the same CP removal and $K$-point DFT processing, the received sensing signal on the $k$th subcarrier is
\begin{equation}
Y_{\mathrm{s}}[k]
=
h_{\mathrm{s}}X[k]
\mathrm{e}^{-\mathrm{j}2\uppi k b_0/K}
+
N_{\mathrm{s}}[k],
\end{equation}
where $h_{\mathrm{s}}$ denotes the complex target reflection coefficient, and $N_{\mathrm{s}}[k]\sim\mathcal{CN}(0,\sigma_{\mathrm{s}}^{2})$ denotes frequency domain additive white Gaussian noise with variance $\sigma_{\mathrm{s}}^2$ at the SRx. Applying matched filtering to $Y_{\mathrm{s}}[k]$, followed by an inverse discrete Fourier transform (IDFT), yields the time domain output at candidate delay index $b$,
\begin{equation}
g_{\mathrm{s}}[b]
=
\sum_{k=0}^{K-1}
X^{*}[k]Y_{\mathrm{s}}[k]
\mathrm{e}^{\mathrm{j}2\uppi kb/K}.
\end{equation}
The integer delay estimate is then given by
\begin{equation}
\widehat{b}_0
=
\arg\max_{b\in\mathcal B}
|g_{\mathrm{s}}[b]|^2,
\end{equation}
where $\mathcal B \subseteq
\mathcal B_{\mathrm{all}}
=
\left\{0,1,\ldots,\frac K2\right\}$ denotes the prescribed delay index search interval.
\subsection{Sensing Performance Metric}
The P-ACF of the transmitted symbols is defined as
\begin{equation}
r[\ell]=\sum_{k=0}^{K-1}|X[k]|^2\mathrm{e}^{\mathrm{j}2\uppi k\ell/K},~\forall \ell\in\mathcal L_K,
\end{equation}
where $\mathcal{L}_K\triangleq\{-K/2+1,\ldots,-1,0,1,\ldots,K/2\}$.
Since $g_{\mathrm{s}}[b]$ can be regarded as a delay-shifted P-ACF,
\begin{equation}
g_{\mathrm{s}}[b]=h_{\mathrm{s}}r[b-b_0]+\widetilde{N}_{\mathrm{s}}[b],
\end{equation}
where $b$ is the candidate delay index and $\widetilde{N}_{\mathrm{s}}[b]$ is the noise term after matched filtering. Therefore, the P-ACF shape directly affects ranging performance. High sidelobes may generate spurious peaks, which can mask weak targets in the presence of strong targets and increase the false alarm probability. Hence, lower sidelobes generally lead to better ranging performance. However, the P-ACF is itself a random function, and its impact on ranging performance must be characterized through statistical properties.
\subsubsection{Mean-Square P-ACF}
We first consider the MS P-ACF. Substituting the P-D OFDM signal model into the P-ACF expression yields
\begin{equation}\label{p-acf}
r[\ell]
=
\sum_{k\in\mathcal{S}_{\mathrm{p}}}
p_{\mathrm{p},k}\mathrm{e}^{\mathrm{j}2\uppi k\ell/K}
+
\sum_{k\in\mathcal{S}_{\mathrm{d}}}
p_{\mathrm{d},k}|s_k|^2\mathrm{e}^{\mathrm{j}2\uppi k\ell/K}.
\end{equation}
The average power of the $k$th subcarrier is defined as
\begin{equation}
w[k]=
\begin{cases}
p_{\mathrm{p},k}, & k\in\mathcal{S}_{\mathrm{p}},\\
p_{\mathrm{d},k}, & k\in\mathcal{S}_{\mathrm{d}}.
\end{cases}
\end{equation}

\begin{proposition}\label{prop1}
The MS P-ACF is given by
\begin{equation}
\mathbb{E}\!\left[|r[\ell]|^2\right] =
\left|\sum_{k=0}^{K-1}w[k]\mathrm{e}^{\mathrm{j}2\uppi k\ell/K}\right|^2 +
(\mu_4-1)\sum_{k\in\mathcal{S}_{\mathrm{d}}}p_{\mathrm{d},k}^2.
\end{equation}
\end{proposition}

\begin{IEEEproof}
See Appendix~\ref{appea}.
\end{IEEEproof}

This expression shows that the MS P-ACF consists of a deterministic term, given by the squared magnitude of the $K$-point DFT of the average subcarrier power spectrum, and a term induced by data randomness. The deterministic term contains the contributions of the pilot powers and average data subcarrier powers together with their coherent cross term; hence, the two contributions are inseparably coupled. Therefore, pilot placement and subcarrier power allocation jointly determine the MS P-ACF and provide design degrees of freedom for sidelobe suppression.

The MS P-ACF mainlobe is
$\mathbb{E}\!\left[|r[0]|^2\right] =
P^2 + (\mu_4-1)\sum_{k\in\mathcal{S}_{\mathrm{d}}}p_{\mathrm{d},k}^2$.
Since the mainlobe varies with power allocation, the normalized MS P-ACF is defined as
\begin{equation}
R_{\mathrm{MS}}[\ell] =
\frac{\mathbb{E}\!\left[|r[\ell]|^2\right]}
{\mathbb{E}\!\left[|r[0]|^2\right]}.
\end{equation}
\subsubsection{ROI-EISL}
In many practical sensing tasks, target delays are constrained by the operational scenario and prior information. We therefore focus on the lag index subset $\Omega_\xi\triangleq\{\ell\in\mathcal{L}_K:1\leq|\ell|\leq\xi\}$. To quantify the MS P-ACF sidelobe level over this region, we define the ROI-EISL as
\begin{equation}
\mathrm{EISL}(\Omega_\xi)
\!=\!
\sum_{\ell\in\Omega_\xi}
\left|\sum_{k=0}^{K-1}w[k]\mathrm{e}^{\mathrm{j}2\uppi k\ell/K}\right|^2
\!+
|\Omega_\xi|(\mu_4-1)
\sum_{k\in\mathcal{S}_{\mathrm{d}}}p_{\mathrm{d},k}^2.
\end{equation}

\begin{corollary}\label{coro1}
When the ROI covers all nonzero lag indices, the global EISL is
\begin{equation}\label{geisl}
\mathrm{EISL}(\Omega_{\mathrm{g}})
=
K\sum_{k\in\mathcal{S}_{\mathrm{p}}}p_{\mathrm{p},k}^2
+
\left[(K-1)\mu_4+1\right]
\sum_{k\in\mathcal{S}_{\mathrm{d}}}p_{\mathrm{d},k}^2
-
P^2,
\end{equation}
where $\Omega_{\mathrm{g}}=\left\{
-\frac K2+1,\ldots,-1,1,\ldots,\frac K2
\right\}$.
\end{corollary}

\begin{IEEEproof}
See Appendix~\ref{appeb}.
\end{IEEEproof}

Similarly, the normalized ROI-EISL is defined as
\begin{equation}
\mathrm{EISL}^{\mathrm{N}}(\Omega_\xi)
=
\frac{\mathrm{EISL}(\Omega_\xi)}
{\mathbb{E}\!\left[|r[0]|^2\right]}.
\end{equation}
\subsubsection{Peak Sidelobe Level Ratio}
While EISL captures the overall sidelobe level, it does not characterize the maximum MS P-ACF sidelobe. We therefore introduce the PSLR as a complementary performance metric, defined as
\begin{equation}
\mathrm{PSLR}=\frac{\displaystyle\max_{\ell\in\Omega_{\mathrm{g}}}\mathbb{E}\!\left[|r[\ell]|^2\right]}{\mathbb{E}\!\left[|r[0]|^2\right]}.
\end{equation}
\section{Effects of P-D Resource Structure on Normalized MS P-ACF Sidelobes}
The impact of subcarrier power allocation on normalized MS P-ACF sidelobes has been analyzed in \cite{Zhang2026OptimalPower}. In this section, we examine how pilot placement and subcarrier power allocation affect these sidelobes.
\subsection{Pilot Placement Invariance under Uniform Power Allocation}
Because UPA is the globally optimal power allocation for the data-only baseline in \cite{Zhang2026OptimalPower}, we first examine the equal power case $p_{\mathrm{p},k}=p_{\mathrm{d},k}=p_0$ for P-D OFDM.

\begin{proposition}\label{prop2}
Under uniform subcarrier power allocation, pilot placement does not affect normalized MS P-ACF sidelobes.
\end{proposition}

\begin{IEEEproof}
For any given $\mathcal{S}_{\mathrm{p}}$ and $\mathcal{S}_{\mathrm{d}}$, substituting UPA into the MS P-ACF yields
\begin{equation}
\mathbb{E}\!\left[|r[\ell]|^2\right]
=
p_0^2\left|\sum_{k=0}^{K-1}\mathrm{e}^{\mathrm{j}2\uppi k\ell/K}\right|^2
+
(\mu_4-1)|\mathcal{S}_{\mathrm{d}}|p_0^2.
\end{equation}
Under full bandwidth transmission, the orthogonality of DFT basis functions gives
\begin{equation}
\sum_{k=0}^{K-1}\mathrm{e}^{\mathrm{j}2\uppi k\ell/K}
=
\begin{cases}
K, & \ell=0,\\
0, & \ell\neq0.
\end{cases}
\end{equation}
Therefore,
\begin{equation}
\mathbb{E}\!\left[|r[\ell]|^2\right]
=
\begin{cases}
p_0^2\left[K^2+(\mu_4-1)|\mathcal{S}_{\mathrm{d}}|\right], & \ell=0,\\
p_0^2(\mu_4-1)|\mathcal{S}_{\mathrm{d}}|, & \ell\neq0.
\end{cases}
\end{equation}
Substituting into the normalized MS P-ACF sidelobe expression gives
\begin{equation}\label{nonpilot}
R_{\mathrm{MS}}[\ell]
=
\frac{(\mu_4-1)|\mathcal{S}_{\mathrm{d}}|}
{K^2+(\mu_4-1)|\mathcal{S}_{\mathrm{d}}|},\qquad
\ell\neq0.
\end{equation}
\eqref{nonpilot} does not depend on the specific elements of the pilot set $\mathcal{S}_{\mathrm{p}}$. Therefore, under uniform subcarrier power allocation, changing pilot placement does not alter normalized MS P-ACF sidelobes.
\end{IEEEproof}

Furthermore, pilot placement affects neither the normalized ROI-EISL nor the PSLR.
\subsection{Optimal P-D OFDM Structure for Normalized Global EISL Minimization}
We next remove the uniform subcarrier power constraint and formulate an optimization problem that minimizes the normalized global EISL to analyze the effects of the P-D resource structure. Let $\beta\triangleq(K-1)\mu_4+1$. The normalized global EISL is
\begin{equation}
\mathrm{EISL}^{\mathrm{N}}(\Omega_{\mathrm{g}})=\frac{K\sum_{k\in\mathcal{S}_{\mathrm{p}}}p_{\mathrm{p},k}^2+\beta\sum_{k\in\mathcal{S}_{\mathrm{d}}}p_{\mathrm{d},k}^2-P^2}{P^2+(\mu_4-1)\sum_{k\in\mathcal{S}_{\mathrm{d}}}p_{\mathrm{d},k}^2}.
\end{equation}
\begin{proposition}\label{prop3}
For fixed $|\mathcal{S}_{\mathrm{p}}|$ and $|\mathcal{S}_{\mathrm{d}}|$, the normalized global EISL is invariant to pilot placement.
\end{proposition}
\begin{IEEEproof}
The expression $\mathrm{EISL}^{\mathrm{N}}(\Omega_{\mathrm{g}})$ does not contain the specific indices of pilot subcarriers. Therefore, pilot placement does not change the value of $\mathrm{EISL}^{\mathrm{N}}(\Omega_{\mathrm{g}})$.
\end{IEEEproof}
The optimization problem therefore depends only on subcarrier powers and is formulated as
\begin{subequations}\label{eislopti}
\begin{align}
\min_{\substack{\{p_{\mathrm{p},k}\}_{k\in\mathcal{S}_{\mathrm{p}}}\\
\{p_{\mathrm{d},k}\}_{k\in\mathcal{S}_{\mathrm{d}}}}}
&\; \mathrm{EISL}^{\mathrm{N}}(\Omega_{\mathrm{g}})
\tag{\theparentequation}\\
\text{s.t.}
&\; \sum_{k\in\mathcal{S}_{\mathrm{p}}}p_{\mathrm{p},k}
+\sum_{k\in\mathcal{S}_{\mathrm{d}}}p_{\mathrm{d},k}=P,
\label{eislopti:a}\\
&\; p_{\mathrm{p},k}\geq0,\; k\in\mathcal{S}_{\mathrm{p}},
\label{eislopti:b}\\
&\; p_{\mathrm{d},k}\geq0,\; k\in\mathcal{S}_{\mathrm{d}}.
\label{eislopti:c}
\end{align}
\end{subequations}
Here, $\mathcal{S}_{\mathrm{p}}$ and $\mathcal{S}_{\mathrm{d}}$ are arbitrary but fixed pilot and data sets. Since the objective function is continuous in the power variables and the feasible set is a nonempty, closed, and bounded simplex, an optimal solution exists.

\begin{theorem}[Optimal solution to normalized global EISL minimization]\label{theoremgeisl}
The optimal power allocation has a two-level power structure. The power allocated to each data subcarrier is
\begin{equation}
p_{\mathrm{d},k}^\star=p_{\mathrm{d}}^\star=\frac{Pt^\star}{|\mathcal{S}_{\mathrm{d}}|},\; k\in\mathcal{S}_{\mathrm{d}},
\end{equation}
and the power allocated to each pilot subcarrier is
\begin{equation}
p_{\mathrm{p},k}^\star=p_{\mathrm{p}}^\star=\frac{P(1-t^\star)}{|\mathcal{S}_{\mathrm{p}}|},\; k\in\mathcal{S}_{\mathrm{p}},
\end{equation}
where
\begin{equation}
\begin{aligned}
t^\star
&=
\frac{2|\mathcal{S}_{\mathrm d}|}
{\zeta+\sqrt{\zeta^2+4|\mathcal{S}_{\mathrm d}|(\mu_4-1)}},\\
\zeta
&=
|\mathcal{S}_{\mathrm d}|+1
+\mu_4\bigl(|\mathcal{S}_{\mathrm p}|-1\bigr).
\end{aligned}
\label{tstar}
\end{equation}
\eqref{tstar} denotes the fraction of total power allocated to data subcarriers. Moreover, $p_{\mathrm{p}}^\star\geq p_{\mathrm{d}}^\star$, with equality if and only if $\mu_4=1$.
\end{theorem}

\begin{IEEEproof}
See Appendix~\ref{appec}.
\end{IEEEproof}

Data subcarriers incur an additional fourth-order moment term in the normalized MS P-ACF. Consequently, the optimal allocation assigns higher power per pilot subcarrier. However, allocating all power to pilots would concentrate the average power spectrum and increase its deterministic sidelobes. When pilots are removed, the optimal solution reduces to Theorem~1 in \cite{Zhang2026OptimalPower}, for which UPA is optimal. Unlike Theorem~2 in \cite{Zhang2026OptimalPower}, the above two-level allocation does not minimize the normalized MS P-ACF sidelobe at every $\ell$. Under nonuniform power allocation, changing pilot placement rearranges the power levels in the frequency domain and therefore changes the MS P-ACF shape.
\subsection{PSLR Minimization under a Two-Level Power Structure}
Although pilot placement does not affect the normalized global EISL, it can reshape the sidelobe distribution. Joint PSLR minimization over pilot placement and arbitrary subcarrier powers is a mixed-integer problem with no available closed-form solution. We therefore isolate the pilot placement effect by considering a fixed two-level power structure.

Let $p_{\mathrm{p}}$ and $p_{\mathrm{d}}$ denote the two-level powers of pilot and data subcarriers, respectively, and define their difference as $\Delta=p_{\mathrm{p}}-p_{\mathrm{d}}$. Then,
\begin{equation}
\sum_{k=0}^{K-1}w[k]\mathrm{e}^{\mathrm{j}2\uppi k\ell/K}
=
p_{\mathrm{d}}\sum_{k=0}^{K-1}\mathrm{e}^{\mathrm{j}2\uppi k\ell/K}
+
\Delta\sum_{k\in\mathcal{S}_{\mathrm{p}}}\mathrm{e}^{\mathrm{j}2\uppi k\ell/K}.
\end{equation}
By the orthogonality of DFT basis functions, the MS P-ACF sidelobe is
\begin{equation}
\mathbb{E}\!\left[|r[\ell]|^2\right]
=
\Delta^2\left|\sum_{k\in\mathcal{S}_{\mathrm{p}}}\mathrm{e}^{\mathrm{j}2\uppi k\ell/K}\right|^2
+
(\mu_4-1)|\mathcal{S}_{\mathrm{d}}|p_{\mathrm{d}}^2,\qquad \ell\neq0.
\end{equation}
Thus, the P-D coupling in the deterministic term disappears. Since $\mathbb{E}\!\left[|r[0]|^2\right]=P^2+(\mu_4-1)|\mathcal{S}_{\mathrm{d}}|p_{\mathrm{d}}^2$, PSLR minimization can be written as
\begin{equation}
\min_{\mathcal{S}_{\mathrm{p}}}
\frac{\Delta^2\displaystyle\max_{\ell\in\Omega_{\mathrm{g}}}\left|\sum_{k\in\mathcal{S}_{\mathrm{p}}}\mathrm{e}^{\mathrm{j}2\uppi k\ell/K}\right|^2+(\mu_4-1)|\mathcal{S}_{\mathrm{d}}|p_{\mathrm{d}}^2}{P^2+(\mu_4-1)|\mathcal{S}_{\mathrm{d}}|p_{\mathrm{d}}^2}.
\end{equation}
This reduces to
\begin{equation}
\min_{\mathcal{S}_{\mathrm{p}}}\max_{\ell\in\Omega_{\mathrm{g}}}\left|\sum_{k\in\mathcal{S}_{\mathrm{p}}}\mathrm{e}^{\mathrm{j}2\uppi k\ell/K}\right|^2.
\end{equation}
Therefore, under the two-level power structure, PSLR minimization is equivalent to pilot placement design: selecting a set of $|\mathcal{S}_{\mathrm{p}}|$ pilot subcarriers to minimize the maximum sidelobe level generated over $\Omega_{\mathrm{g}}$.

Ideal pilot placement structures have been analyzed in \cite{Bouziane2026OptimizedPilot}. By Parseval's theorem, the aggregate sidelobe level of the pilot placement satisfies
\begin{equation}
\sum_{\ell\in\Omega_{\mathrm{g}}}
\left|\sum_{k\in\mathcal{S}_{\mathrm{p}}}\mathrm{e}^{\mathrm{j}2\uppi k\ell/K}\right|^2
=
|\mathcal{S}_{\mathrm{p}}|\left(K-|\mathcal{S}_{\mathrm{p}}|\right).
\end{equation}
Therefore, the maximum pilot placement sidelobe level is lower bounded by
\begin{equation}
\max_{\ell\in\Omega_{\mathrm{g}}}
\left|\sum_{k\in\mathcal{S}_{\mathrm{p}}}\mathrm{e}^{\mathrm{j}2\uppi k\ell/K}\right|^2
\geq
\frac{|\mathcal{S}_{\mathrm{p}}|\left(K-|\mathcal{S}_{\mathrm{p}}|\right)}{K-1}.
\end{equation}
The lower bound is attained only when the squared pilot pattern sidelobe magnitudes are equal over $\Omega_{\mathrm{g}}$. Furthermore, define the cyclic difference multiplicity as
\begin{equation}
N_{\mathcal{S}_{\mathrm{p}}}[d]
=
\left|\left\{(k_i,k_j)\in\mathcal{S}_{\mathrm{p}}^2:
k_i-k_j\equiv d\pmod K\right\}\right|,
\end{equation}
where $N_{\mathcal{S}_{\mathrm{p}}}[d]$ denotes the number of ordered pairs of pilot indices whose difference modulo $K$ equals $d$. If there exists a positive integer $\lambda$ such that $N_{\mathcal{S}_{\mathrm{p}}}[d]=\lambda$ for every $d\in\{1,2,\ldots,K-1\}$, then $\mathcal{S}_{\mathrm{p}}$ is a cyclic difference set with parameters $\left(K,|\mathcal{S}_{\mathrm{p}}|,\lambda\right)$. It is further shown in \cite{Bouziane2026OptimizedPilot} that equal squared pilot pattern sidelobe magnitudes are equivalent to a cyclic difference set pilot placement. However, cyclic difference sets do not always exist. In particular, let
\begin{equation}
\lambda=\frac{|\mathcal{S}_{\mathrm{p}}|\left(|\mathcal{S}_{\mathrm{p}}|-1\right)}{K-1}.
\end{equation}
If $\lambda$ is not an integer, no cyclic difference set exists.

Although the above discussion only analyzes the structure of ideal solutions to the restricted PSLR minimization problem, it suggests that, under nonuniform power allocation, nonuniform pilot placement provides a potential degree of freedom for reducing normalized MS P-ACF sidelobes.

\section{ROI-Aware P-D Resource Optimization}
We further investigate whether pilot placement and subcarrier power allocation can suppress normalized MS P-ACF sidelobes within the ROI. For data-only OFDM under full bandwidth transmission, UPA minimizes the normalized MS P-ACF sidelobe at every nonzero $\ell$ and hence the normalized global EISL \cite{Zhang2026OptimalPower}. In contrast, the inclusion of pilots provides additional degrees of freedom for sidelobe optimization within the ROI.
\subsection{Normalized ROI-EISL Minimization}
For analytical convenience, define the ROI matrix as
\begin{equation}
[\boldsymbol{C}_{\Omega_\xi}]_{m,n}
=\sum_{\ell\in\Omega_\xi}\mathrm{e}^{\mathrm{j}2\uppi(m-n)\ell/K},
\end{equation}
which is a Hermitian positive semidefinite circulant matrix. Define the subcarrier power vector as $\boldsymbol{w}=[w[0],w[1],\ldots,w[K-1]]^{\mathrm{T}}$ and the data subcarrier indicator matrix as
\begin{equation}
[\boldsymbol{D}_{\mathrm{d}}(\mathcal{S}_{\mathrm{p}})]_{k,k}
=\begin{cases}
1, & k\in\mathcal{S}_{\mathrm{d}},\\
0, & k\in\mathcal{S}_{\mathrm{p}}.
\end{cases}
\end{equation}

\begin{proposition}\label{prop4}
The normalized ROI-EISL can be expressed as
\begin{equation}
\mathrm{EISL}^{\mathrm{N}}(\Omega_\xi)
=
\frac{
\boldsymbol{w}^{\mathrm{T}}\!\left[
\boldsymbol{C}_{\Omega_\xi}
+|\Omega_\xi|(\mu_4-1)\boldsymbol{D}_{\mathrm{d}}(\mathcal{S}_{\mathrm{p}})
\right]\boldsymbol{w}
}{
P^2+(\mu_4-1)\boldsymbol{w}^{\mathrm{T}}
\boldsymbol{D}_{\mathrm{d}}(\mathcal{S}_{\mathrm{p}})\boldsymbol{w}
}.
\end{equation}
\end{proposition}

\begin{IEEEproof}
See Appendix~\ref{apped}.
\end{IEEEproof}

Pilot placement affects the normalized ROI-EISL through the joint structure of $\boldsymbol{D}_{\mathrm{d}}(\mathcal{S}_{\mathrm{p}})$ and the power vector. For any given $\Omega_\xi$, the normalized ROI-EISL minimization problem is formulated as
\begin{subequations}\label{ROIEISL}
\begin{align}
\min_{\boldsymbol{w},\,\mathcal{S}_{\mathrm{p}}}
&\; \mathrm{EISL}^{\mathrm{N}}(\Omega_\xi)
\tag{\theparentequation}\\
\text{s.t.}
&\; \mathbf{1}^{\mathrm{T}}\boldsymbol{w}=P,
\label{ROIEISL:a}\\
&\; \boldsymbol{w}\geq\mathbf{0}.
\label{ROIEISL:b}
\end{align}
\end{subequations}
This is a mixed-integer quadratic fractional program.
\subsection{Globally Optimal Strictly Uniform Pilot Placement}
To model communication requirements on individual data subcarriers, we impose a minimum power constraint on each data subcarrier. We first consider a uniform data subcarrier power lower bound, i.e., $p_{\mathrm{d},k}\geq p_{\mathrm{min}}$ for all $k\in\mathcal{S}_{\mathrm{d}}$.

\begin{proposition}\label{prop5}
The normalized ROI-EISL is lower bounded by
\begin{equation}
\mathrm{EISL}^{\mathrm{N}}(\Omega_\xi)
\geq
\frac{
|\Omega_\xi|(\mu_4-1)|\mathcal{S}_{\mathrm{d}}|p_{\mathrm{min}}^2
}{
P^2+(\mu_4-1)|\mathcal{S}_{\mathrm{d}}|p_{\mathrm{min}}^2
}.
\end{equation}
\end{proposition}

\begin{IEEEproof}
See Appendix~\ref{appee}.
\end{IEEEproof}

\begin{theorem}[Global optimality of strictly uniform pilot placement]\label{theorem2}
Suppose that $|\mathcal{S}_{\mathrm{p}}|$ divides $K$, $P\geq|\mathcal{S}_{\mathrm{d}}|p_{\mathrm{min}}$, and $\xi<|\mathcal{S}_{\mathrm{p}}|$. A strictly uniform pilot placement can be constructed as
\begin{equation}
\mathcal{S}_{\mathrm{p}}^{\star}
=
\left\{
\left(k_0+m\frac{K}{|\mathcal{S}_{\mathrm{p}}|}\right)\bmod K
\;\middle|\; m=0,1,\ldots,|\mathcal{S}_{\mathrm{p}}|-1
\right\},
\end{equation}
where $k_0$ denotes the starting subcarrier. With
\begin{equation}
\begin{aligned}
p_{\mathrm{d},k}^{\star}&=p_{\mathrm{min}}, && k\in\mathcal{S}_{\mathrm{d}}^{\star},\\
p_{\mathrm{p},k}^{\star}&=\frac{P-|\mathcal{S}_{\mathrm{d}}^{\star}|p_{\mathrm{min}}}{|\mathcal{S}_{\mathrm{p}}^{\star}|}, && k\in\mathcal{S}_{\mathrm{p}}^{\star},
\end{aligned}
\end{equation}
the normalized ROI-EISL attains its lower bound.
\end{theorem}

\begin{IEEEproof}
See Appendix~\ref{appef}.
\end{IEEEproof}

This construction is not necessarily unique. Any feasible pilot placement and power allocation satisfying $\boldsymbol{w}^{\mathrm{T}}\boldsymbol{C}_{\Omega_\xi}\boldsymbol{w}=0$ and 
$\boldsymbol{w}^{\mathrm{T}}\boldsymbol{D}_{\mathrm{d}}(\mathcal{S}_{\mathrm{p}})\boldsymbol{w}
=|\mathcal{S}_{\mathrm{d}}|p_{\mathrm{min}}^2$ achieves the same global optimum.

\begin{remark}
More generally, the same construction attains the lower bound for any ROI $\Omega$ satisfying $\ell\not\equiv0\pmod{|\mathcal{S}_{\mathrm{p}}|}$ for every $\ell\in\Omega$. The resulting MS P-ACF has a comb-like structure, and its sidelobes attain the minimum value at all indices satisfying this congruence condition. The optimality result relies on the uniform data subcarrier power lower bound, which yields the required two-level structure; it does not generally extend to nonuniform data subcarrier power lower bounds.
\end{remark}
 
\subsection{Bisection-and-Branch-and-Bound Algorithm for Normalized ROI-EISL Minimization}
For the more general case, we solve the normalized ROI-EISL minimization problem. 
\subsubsection{Problem Analysis}
Since no sidelobe of any normalized MS P-ACF exceeds its mainlobe, the normalized ROI-EISL satisfies $0 \leq \mathrm{EISL}^{\mathrm{N}}(\Omega_\xi) \leq |\Omega_\xi|$. For any candidate objective value $\eta\in[0,|\Omega_\xi]$, define
\begin{equation}
\boldsymbol{A}_{\eta}
=
\boldsymbol{C}_{\Omega_\xi}
+(\mu_4-1)(|\Omega_\xi|-\eta)
\boldsymbol{D}_{\mathrm{d}}(\mathcal{S}_{\mathrm{p}}),
\end{equation}
and consider the mixed-integer quadratic programming (MIQP) problem
\begin{subequations}\label{miqp}
\begin{align}
q(\eta)=\min_{\boldsymbol{w}}
&\; \boldsymbol{w}^{\mathrm{T}}\boldsymbol{A}_{\eta}\boldsymbol{w}
\tag{\theparentequation}\\
\text{s.t.}
&\; \eqref{ROIEISL:a},\ \eqref{ROIEISL:b},\ \text{and}
\notag\\
&\; w[k]\geq p_{\mathrm{min},k},\; k\in\mathcal{S}_{\mathrm{d}}.
\label{paramqp:a}
\end{align}
\end{subequations}
Here, $p_{\mathrm{min},k}$ is the prescribed power lower bound if the $k$th subcarrier is assigned to data transmission.

\begin{proposition}\label{prop6}
For a fixed pilot set $\mathcal{S}_{\mathrm{p}}$, the continuous power allocation subproblem of \eqref{ROIEISL} is quasiconvex, and its optimal value is given by
\begin{equation}
F^{\star}
=
\min\left\{
\eta\in[0,|\Omega_\xi|]:
q(\eta)\leq\eta P^2
\right\}.
\end{equation}
It can be obtained by bisection over $\eta\in[0,|\Omega_\xi|]$.
\end{proposition}

\begin{IEEEproof}
See Appendix~\ref{appeg}.
\end{IEEEproof}

We now reconsider $\mathcal{S}_{\mathrm{p}}$. In principle, \eqref{ROIEISL} can be solved by enumerating all pilot placements, applying Proposition~\ref{prop6} to obtain the corresponding optimal value, and selecting the pilot placement with the smallest optimum. Direct enumeration requires $\binom{K}{|\mathcal{S}_{\mathrm{p}}|}$ candidates, which rapidly becomes prohibitive as $K$ increases. Greedy and exchange-based heuristics can produce feasible solutions but generally cannot guarantee global optimality. This motivates the bisection-and-branch-and-bound algorithm described next.

\subsubsection{Bisection-and-Branch-and-Bound Algorithm}
To facilitate optimization, we explicitly incorporate pilot placement into the decision variables by defining the binary indicator
\begin{equation}
x_k=
\begin{cases}
1, & k\in\mathcal{S}_{\mathrm{p}},\\
0, & k\in\mathcal{S}_{\mathrm{d}},
\end{cases},
\end{equation}
We further introduce the pilot and data power variables $u_k=p_{\mathrm{p},k}$ for $x_k=1$ and $v_k=p_{\mathrm{d},k}$ for $x_k=0$. The variable $u_k$ is nonzero only on pilot subcarriers, whereas $v_k$ is nonzero only on data subcarriers. Define
$\boldsymbol{u}=[u_0,u_1,\ldots,u_{K-1}]^{\mathrm{T}}$,
$\boldsymbol{v}=[v_0,v_1,\ldots,v_{K-1}]^{\mathrm{T}}$, and
$\boldsymbol{x}=[x_0,x_1,\ldots,x_{K-1}]^{\mathrm{T}}$.
For compactness, define the objective function as
\begin{equation}\label{jointObjective}
\mathcal{J}(\boldsymbol{u},\boldsymbol{v})
=
\frac{
(\boldsymbol{u}+\boldsymbol{v})^{\mathrm{T}}
\boldsymbol{C}_{\Omega_\xi}(\boldsymbol{u}+\boldsymbol{v})+|\Omega_\xi|(\mu_4-1)\boldsymbol{v}^{\mathrm{T}}\boldsymbol{v}
}{
P^2+(\mu_4-1)\boldsymbol{v}^{\mathrm{T}}\boldsymbol{v}
}.
\end{equation}
Then, \eqref{ROIEISL} can be equivalently written as
\begin{subequations}\label{jointROIEISL}
\begin{align}
\min_{\boldsymbol{u},\,\boldsymbol{v},\,\boldsymbol{x}}
&\; \mathcal{J}(\boldsymbol{u},\boldsymbol{v})
\tag{\theparentequation}\\
\text{s.t.}
&\; \mathbf{1}^{\mathrm{T}}(\boldsymbol{u}+\boldsymbol{v})=P,
\label{jointROIEISL:a}\\
&\; \sum_{k=0}^{K-1}x_k=|\mathcal{S}_{\mathrm{p}}|,
\label{jointROIEISL:b}\\
&\; 0\leq u_k\leq Px_k,\; \forall k\in\mathcal{S},
\label{jointROIEISL:c}\\
&\; p_{\mathrm{min},k}(1-x_k)\leq v_k\leq P(1-x_k),\; \forall k\in\mathcal{S},
\label{jointROIEISL:d}\\
&\; x_k\in\{0,1\},\; \forall k\in\mathcal{S}.
\label{jointROIEISL:e}
\end{align}
\end{subequations}

We first perform outer bisection. At the $i$th iteration, let
\begin{equation}
\eta_i=\frac{\eta_{\mathrm{L}}+\eta_{\mathrm{U}}}{2},
\end{equation}
where $\eta_{\mathrm{L}}$ and $\eta_{\mathrm{U}}$ are the current lower and upper bounds on the optimal value, respectively. For a given $\eta_i$, define
\begin{equation}\label{phiEta}
\Phi_{\eta_i}(\boldsymbol{u},\boldsymbol{v})
=
(\boldsymbol{u}+\boldsymbol{v})^{\mathrm{T}}
\boldsymbol{C}_{\Omega_\xi}(\boldsymbol{u}+\boldsymbol{v})+(\mu_4-1)(|\Omega_\xi|-\eta_i)\boldsymbol{v}^{\mathrm{T}}\boldsymbol{v}.
\end{equation}
The joint feasibility test can then be formulated as
\begin{equation}\label{bigg}
\begin{aligned}
G(\eta_i)=\min_{\boldsymbol{u},\,\boldsymbol{v},\,\boldsymbol{x}}
&\; \Phi_{\eta_i}(\boldsymbol{u},\boldsymbol{v}),\\
\text{s.t.}
&\; \eqref{jointROIEISL:a}\text{--}\eqref{jointROIEISL:e}.
\end{aligned}
\end{equation}
Following the proof in Appendix~\ref{appeg}, determining whether there exist $\boldsymbol{u}$, $\boldsymbol{v}$, and $\boldsymbol{x}$ such that $\mathcal{J}(\boldsymbol{u},\boldsymbol{v})\leq\eta_i$ is equivalent to checking whether $G(\eta_i)\leq\eta_iP^2$. If this condition holds, we update $\eta_{\mathrm{U}}=\eta_i$; otherwise, we update $\eta_{\mathrm{L}}=\eta_i$. This procedure continues until the prescribed accuracy is attained.

For each $\eta_i$, problem \eqref{bigg} is solved using branch-and-bound (BnB). Each node $\mathcal{N}$ represents a partial assignment of $\boldsymbol{x}$, where $\mathcal{F}_1(\mathcal{N})$, $\mathcal{F}_0(\mathcal{N})$, and $\mathcal{F}(\mathcal{N})$ denote the indices fixed as pilots, fixed as data subcarriers, and unfixed, respectively. Each node is processed by continuous relaxation, integer-solution checking, rounding, and branching as follows.

\begin{enumerate}
\renewcommand{\labelenumi}{(\alph{enumi})}
\item \emph{Continuous relaxation and pruning:} Relax \eqref{jointROIEISL:e} to $0\leq x_k\leq1$ and solve the resulting convex QP, yielding the lower bound $L_{\mathcal{N}}$ and solution $(\boldsymbol{u}^{\mathrm{rel}},\boldsymbol{v}^{\mathrm{rel}},\boldsymbol{x}^{\mathrm{rel}})$. If $L_{\mathcal{N}}>\eta_iP^2$, node $\mathcal{N}$ is pruned.

\item \emph{Integer solution checking:} If $L_{\mathcal{N}}\leq\eta_iP^2$ and $\boldsymbol{x}^{\mathrm{rel}}$ is integral, a feasible solution to \eqref{bigg} is obtained and BnB terminates.

\item \emph{Rounding:} If $\boldsymbol{x}^{\mathrm{rel}}$ is fractional, retain the pilots in $\mathcal{F}1(\mathcal{N})$ and select the unfixed indices with the largest $x_k^{\mathrm{rel}}$ until $|\mathcal{S}{\mathrm{p}}|$ pilots are obtained. Fix the resulting $\widehat{\boldsymbol{x}}$ and reoptimize $\boldsymbol{u}$ and $\boldsymbol{v}$; if $\Phi_{\eta_i}(\boldsymbol{u},\boldsymbol{v})\leq\eta_iP^2$, BnB terminates.

\item \emph{Branching:}
Otherwise, select
\begin{equation}
k^{\star}
=
\arg\min_{\substack{k\in\mathcal{F}(\mathcal{N})\\
0<x_k^{\mathrm{rel}}<1}}
\left|x_k^{\mathrm{rel}}-\frac{1}{2}\right|,
\end{equation}
and generate two child nodes with $x_{k^\star}=1$ and $x_{k^\star}=0$.

\end{enumerate}

\begin{remark}
If the BnB termination criterion is reached before feasibility at $\eta_i$ can be certified, the algorithm terminates the outer bisection and reports the best feasible incumbent together with its actual objective value as an approximate solution. When $(\mu_4-1)\boldsymbol{v}^{\mathrm{T}}\boldsymbol{v}\ll P^2$, minimizing the unnormalized ROI-EISL can be used as a lower-complexity approximation.
\end{remark}

\subsubsection{Complexity Analysis}
The number of bisection iterations required to attain the target accuracy
$\varepsilon$ is $N_{\eta}
=
\max\left\{
0,
\left\lceil
\log_2\frac{|\Omega_\xi|}{\varepsilon}
\right\rceil
\right\}$. The BnB tree contains $N_{\mathrm{node}}=O(2^K)$ nodes in the worst case. Solving the convex QP at each node using an interior-point method has complexity $O\!\left(K^{3.5}\log\frac{1}{\varepsilon_{\mathrm{c}}}\right)$, where $\varepsilon_{\mathrm{c}}$ is the accuracy of the continuous QP solver. Therefore, the worst-case complexity is $O\!\left(
N_{\eta}2^K K^{3.5}
\log\frac{1}{\varepsilon_{\mathrm{c}}}
\right)$.

\begin{figure*}[!t]
 \centering
 \includegraphics[width=\textwidth]
 {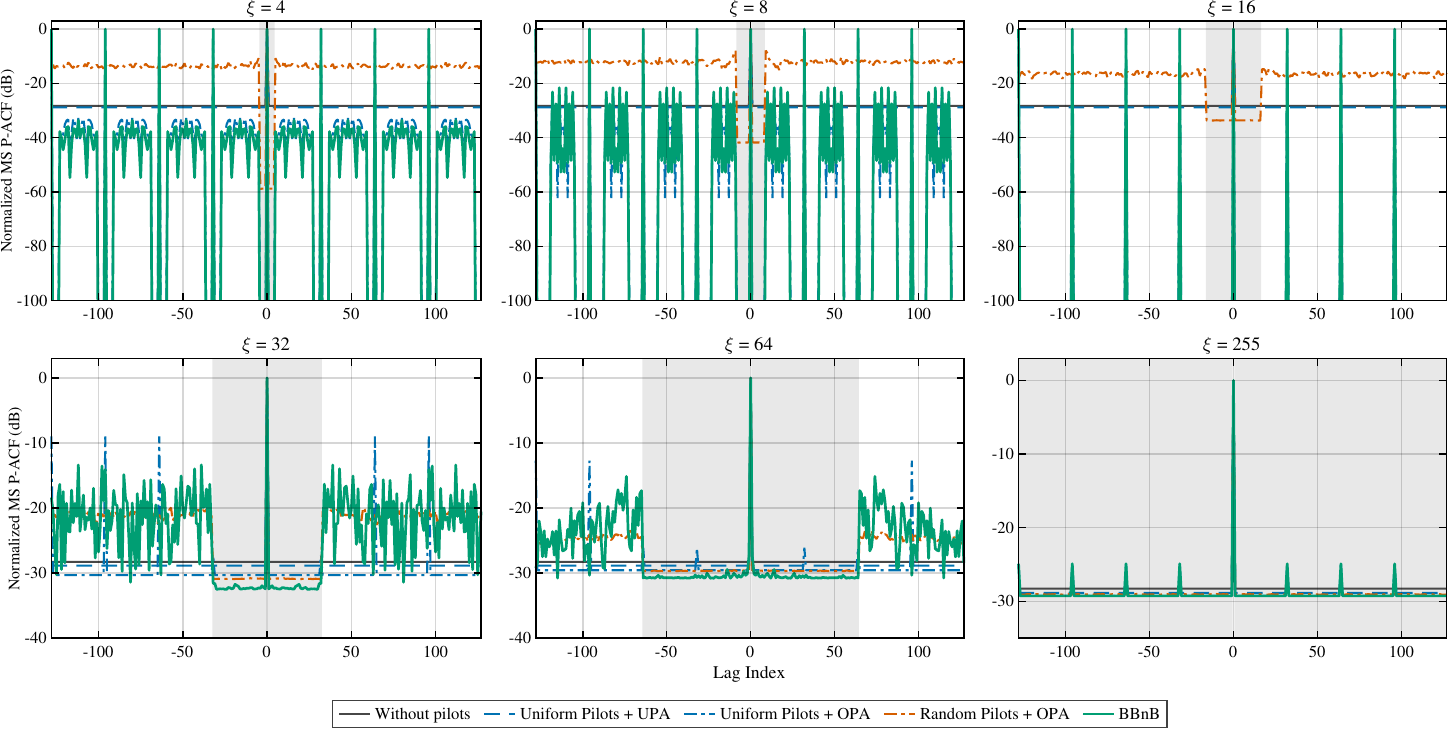}
 \caption{Normalized MS P-ACFs of different P-D OFDM configurations for varying ROI sizes without data subcarrier power constraints.}
 \label{fig1}
\end{figure*}

\section{Numerical Results}
This section presents numerical results to validate the theoretical analysis and optimization strategy. Unless otherwise specified, the baseline configuration uses $K=256$, a pilot fraction of $12.5\%$, $P=1$, and 64-QAM modulation. Results involving random pilot placements or random power lower bounds are averaged over 100 Monte Carlo trials, except where explicitly noted.

Without data subcarrier power constraints, Fig.~\ref{fig1} compares the normalized MS P-ACFs of BBnB with four baseline P-D OFDM configurations over different ROI sizes. Here, optimal power allocation (OPA) denotes power optimization for a fixed pilot placement. The four baseline schemes are
\begin{itemize}
 \item ``Without Pilots,'' which employs optimal power allocation with data-only transmission.
 \item ``Uniform Pilots + UPA,'' which uses uniform pilot placement and uniform power allocation across all subcarriers.
 \item ``Uniform Pilots + OPA,'' which uses uniform pilot placement and optimal power allocation.
 \item ``Random Pilots + OPA,'' which uses randomly placed pilots and optimal power allocation.
\end{itemize}

\begin{figure}[h]
\centering
\includegraphics[width=3.4in]{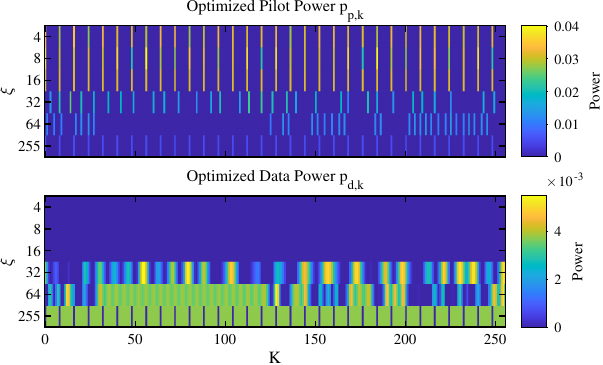}
\caption{Pilot placement and subcarrier power allocation obtained by BBnB for varying ROI sizes without data subcarrier power constraints.}
\label{fig2}
\end{figure}

The BBnB solution illustrates the joint effects of pilot placement and subcarrier power allocation on normalized MS P-ACF sidelobes. As the ROI shrinks, the optimized normalized MS P-ACF sidelobes become lower within the ROI. The result for ``Uniform Pilots + OPA'' at $\xi=16$ is consistent with Theorem~\ref{theorem2}: because $|\mathcal{S}_{\mathrm{p}}|=32$ and $\xi<32$, a strictly uniform pilot placement with a two-level power structure is globally optimal. The results for smaller ROIs, together with the pilot placements and power allocations in Fig.~\ref{fig2}, also indicate that this construction need not be unique. For $\xi=32$, BBnB achieves the lowest sidelobe level within the ROI among the compared schemes, compared with ``Uniform Pilots + OPA'', it reduces the ROI sidelobe level by up to 2.19 dB. Fig.~\ref{fig2} depicts the variations in the P-D resource configuration of the BBnB solution as the ROI expands. The pilot placement transitions from uniform to nonuniform and then back to uniform, while the fraction of total power allocated to pilots gradually decreases. This is because grating lobes induced by pilots become increasingly significant as the optimization region expands, requiring a larger fraction of the transmit power to be allocated to data subcarriers to reduce the sidelobe level.

To compare different ROI sizes, we report the region of interest expected sidelobe level (ROI-ESL), defined as the normalized ROI-EISL per lag, $\mathrm{ROI\mbox{-}ESL}\triangleq\mathrm{EISL}^{\mathrm{N}}(\Omega_\xi)/|\Omega_\xi|$.

\begin{figure}[h]
\centering
\includegraphics[width=3.3in]{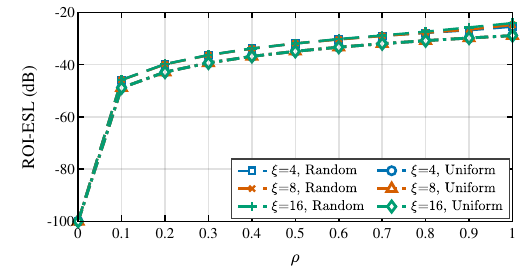}
\caption{ROI-ESL versus the data subcarrier power lower bound coefficient $\rho$ under uniform and random data subcarrier power lower bounds for different ROI sizes.}
\label{fig3}
\end{figure}

\begin{figure}[h]
\centering
\includegraphics[width=3.3in]{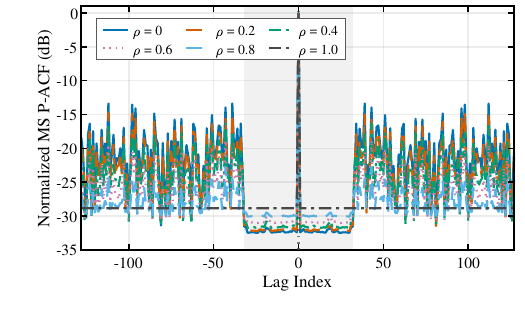}
\caption{Normalized MS P-ACFs obtained by BBnB for different data subcarrier power lower bound coefficients $\rho$ with $\xi=32$.}
\label{fig4}
\end{figure}

\begin{figure}[h]
\centering
\includegraphics[width=3.3in]{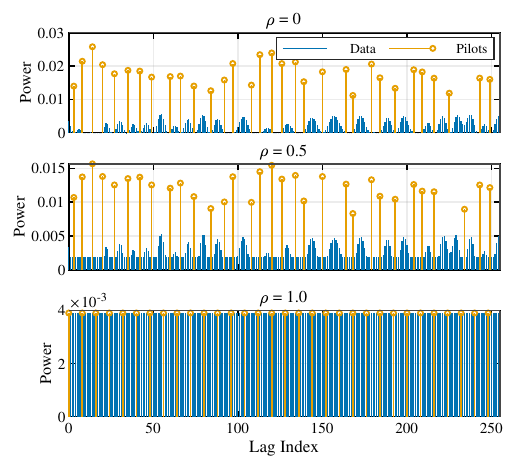}
\caption{Pilot placement and subcarrier power allocation obtained by BBnB for different data subcarrier power lower bound coefficients $\rho$ with $\xi=32$.}
\label{fig5}
\end{figure}

\begin{figure}[h]
\centering
\includegraphics[width=3.3in]{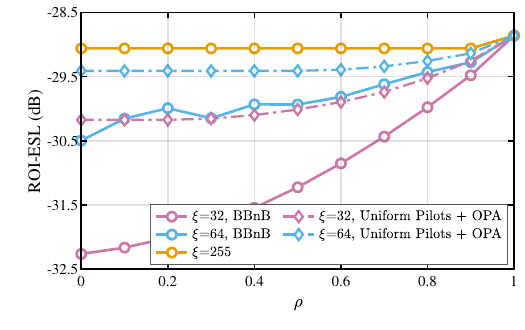}
\caption{ROI-ESL versus the data subcarrier power lower bound coefficient $\rho$ for BBnB and ``Uniform Pilots + OPA'' under uniform data subcarrier power lower bounds.}
\label{fig6}
\end{figure}

Next, we impose data subcarrier power constraints. Fig.~\ref{fig3} compares random lower bound coefficients (``Random'') with a common lower bound coefficient (``Uniform'') as $\rho=Kp_{\mathrm{min}}/P$ increases; the random coefficients are generated with mean $\rho$. A larger $\rho$ imposes stronger data power requirements and increases the ROI-ESL, thereby revealing a sensing--communication tradeoff. Values below the numerical display floor are plotted at $-100~\mathrm{dB}$. The comparison between ``Random'' and ``Uniform,'' together with the results for $\xi=4,8,16$ under ``Uniform,'' is consistent with Theorem~\ref{theorem2}. Fig.~\ref{fig4} shows the normalized MS P-ACF sidelobes within the ROI obtained by BBnB as $\rho$ increases for $\xi=32$, and Fig.~\ref{fig5} shows the corresponding pilot placements and power allocations. Fig.~\ref{fig6} shows that BBnB achieves a lower ROI-ESL than ``Uniform Pilots + OPA'' over the considered range of $\rho$.

\begin{figure}[!h]
\centering
\includegraphics[width=3.3in]{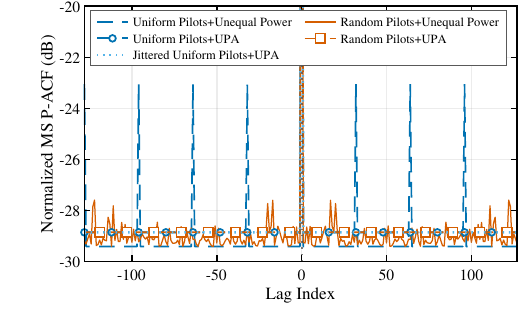}
\caption{Normalized MS P-ACFs for different pilot placements under uniform and nonuniform subcarrier power allocation.}
\label{fig7}
\end{figure}

Fig.~\ref{fig7} confirms the pilot placement invariance under UPA established in Proposition~\ref{prop2}, while different pilot placements yield distinct sidelobe shapes under nonuniform power allocation.


\begin{table}[!h]
 \centering
 \caption{Comparison of sensing performance for different pilot placements}
 \label{tab2}
 \begin{tabular}{|c|c|c|}
 \hline
 & Normalized global EISL & PSLR \\ \hline
 Uniform Pilots & -4.99 dB & -24.90 dB \\ \hline
 Random Pilots & -4.99 dB & -28.22 dB \\ \hline
 Jittered Uniform Pilots & -4.99 dB & -28.31 dB \\ \hline
 \end{tabular}
\end{table}

Table~\ref{tab2} confirm the pilot placement invariance of the normalized global EISL established in Section~III-B, while showing substantial differences in PSLR. Here, ``Jittered Uniform Pilots'' denotes a baseline obtained by perturbing the positions of a strictly uniform pilot placement.

\begin{figure}[h]
\centering
\includegraphics[width=3.3in]{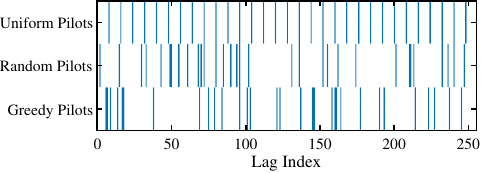}
\caption{Pilot placements for the ``Greedy Pilots,'' ``Random Pilots,'' and ``Uniform Pilots'' schemes under the two-level power structure.}
\label{fig9}
\end{figure}

\begin{figure}[h]
\centering
\includegraphics[width=3.4in]{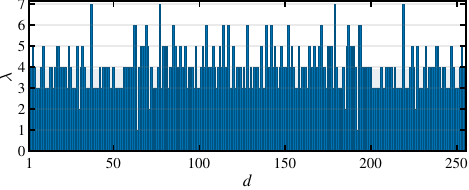}
\caption{Cyclic difference multiplicities of ``Greedy Pilots'' and the equal multiplicity value $\lambda$ associated with a cyclic difference set.}
\label{fig10}
\end{figure}

\begin{figure}[h]
\centering
\includegraphics[width=3.4in]{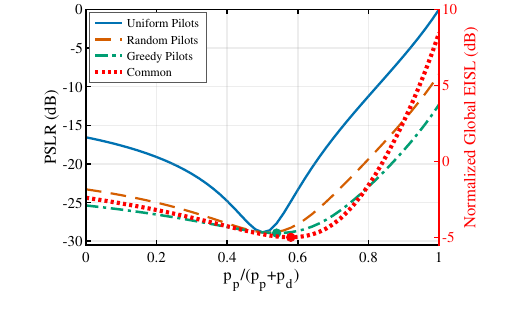}
\caption{PSLR and normalized global EISL versus $p_{\mathrm p}/(p_{\mathrm p}+p_{\mathrm d})$ for different pilot placements under the two-level power structure.}
\label{fig11}
\end{figure}

We use a greedy algorithm to obtain ``Greedy Pilots'' as an approximation to the cyclic difference set structure discussed in Section~III-C. Fig.~\ref{fig9} compares ``Greedy Pilots'' with ``Random Pilots'' (without Monte Carlo averaging) and ``Uniform Pilots.'' Fig.~\ref{fig10} shows the cyclic difference multiplicities of ``Greedy Pilots.'' The equal multiplicity reference is $\lambda=3.89$, and the mean-square error relative to this reference is $1.03$. Fig.~\ref{fig11} shows the PSLR and normalized global EISL as $p_\mathrm{p}/(p_\mathrm{p}+p_\mathrm{d})$ varies. ``Greedy Pilots'' achieves a lower PSLR than the other schemes over the considered range. The curve labeled ``Common'' denotes the normalized global EISL, which is independent of pilot placement by Proposition~\ref{prop3}. For ``Uniform Pilots'' and ``Random Pilots,'' the PSLR is minimized at $p_\mathrm{p}=p_\mathrm{d}$, corresponding to UPA. For ``Greedy Pilots,'' the value of $p_\mathrm{p}/(p_\mathrm{p}+p_\mathrm{d})$ minimizing the PSLR is closer to that minimizing the normalized global EISL.

\begin{figure}[h]
\centering
\includegraphics[width=3.4in]{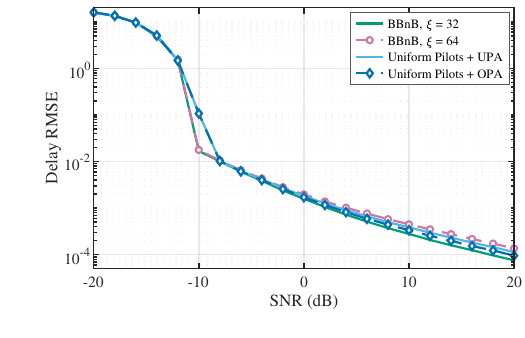}
\caption{Delay RMSE versus sensing SNR for BBnB with $\xi\in\{32,64\}$, ``Uniform Pilots + UPA,'' and ``Uniform Pilots + OPA.''}
\label{fig12}
\end{figure}

Fig.~\ref{fig12} shows the delay root-mean-square error (RMSE) versus sensing SNR. After integer-delay detection, three-point parabolic interpolation is used for fractional-delay estimation. For a fair comparison, the same search window, $[-32,-1]\cup[1,32]$, is used for all schemes. The BBnB scheme with $\xi=32$ achieves the lowest delay RMSE. Compared with ``Uniform Pilots + OPA'', BBnB reduces the delay RMSE by up to 84.3$\%$ at an SNR of $-10$ dB, demonstrating that lower sidelobe levels can translate into higher ranging accuracy.

\begin{figure}[h]
\centering
\includegraphics[width=3.4in]{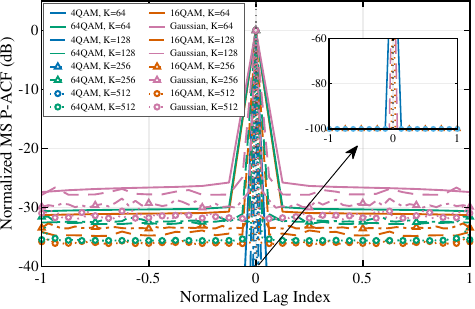}
\caption{Normalized MS P-ACFs obtained by BBnB for different QAM orders and numbers of subcarriers without data subcarrier power constraints.}
\label{fig13}
\end{figure}

Next, we evaluate the normalized MS P-ACFs for different QAM orders and numbers of subcarriers, without imposing data subcarrier power constraints. As shown in Fig.~\ref{fig13}, the horizontal axis is the lag index normalized by $K/8$, and only the interval $|\ell|\leq K/8$ is displayed. As the QAM order increases, the sidelobe level increases, while circularly symmetric complex Gaussian signaling yields the highest level. For 4-QAM, the zero sidelobe result is consistent with Proposition~\ref{prop2} because $\mu_4=1$. Increasing the number of subcarriers reduces the sidelobes, indicating that symbol power fluctuations are averaged over more subcarriers.

\begin{figure}[h]
\centering
\includegraphics[width=3.4in]{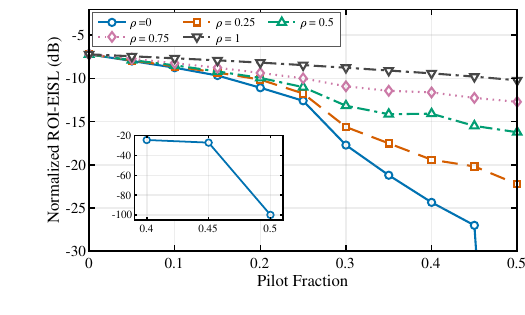}
\caption{Normalized ROI-EISL versus pilot fraction for different data subcarrier power lower bounds with $\xi=64$.}
\label{fig14}
\end{figure}

\begin{figure}[h]
\centering
\includegraphics[width=3.4in]{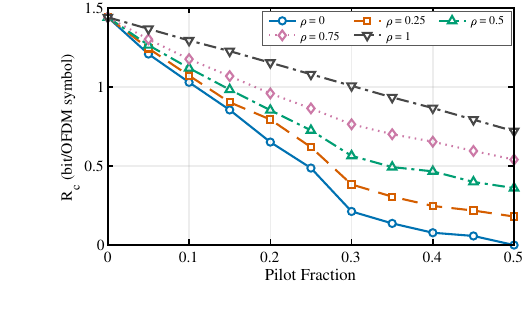}
\caption{Communication rate per OFDM symbol versus pilot fraction for different data subcarrier power lower bounds under a unit-gain flat channel.}
\label{fig15}
\end{figure}

Finally, we evaluate sensing and communication performance for different pilot fractions with $\xi=64$. Fig.~\ref{fig14} shows that the normalized ROI-EISL generally decreases as the pilot fraction increases. Under a unit-gain flat channel, $|H_{\mathrm c}[k]|^2=1$, and unit noise power, $\sigma_{\mathrm c}^2=1$, Fig.~\ref{fig15} shows that the communication rate decreases as more subcarriers are assigned to pilots. Thus, the pilot fraction controls the tradeoff between normalized ROI-EISL and communication rate.

\section{Conclusion}
This paper analyzed and optimized the ranging performance of P-D OFDM signals using the normalized MS P-ACF. We derived a closed-form MS P-ACF, defined its normalized form, and characterized how pilot placement and subcarrier power allocation shape the sidelobes. We then analyzed normalized global EISL minimization and PSLR minimization under a two-level power structure, and formulated the joint pilot placement and power allocation problem for normalized ROI-EISL minimization. Under a uniform data subcarrier power lower bound and $\xi<|\mathcal{S}_{\mathrm{p}}|$, a globally optimal strictly uniform pilot placement was obtained; a BBnB algorithm was developed for the general case. Numerical results validated the analysis and the proposed optimization. Future work will incorporate channel state information and other practical system parameters, and extend the framework to velocity estimation.

\appendices
\section{Proof of Proposition~\ref{prop1}}\label{appea}
By rearranging \eqref{p-acf}, we obtain
\begin{equation}
r[\ell]=\sum_{k=0}^{K-1}w[k]\mathrm{e}^{\mathrm{j}2\uppi k\ell/K}
+\sum_{k\in\mathcal{S}_{\mathrm{d}}}p_{\mathrm{d},k}(|s_k|^2-1)\mathrm{e}^{\mathrm{j}2\uppi k\ell/K}.
\end{equation}
Define
\begin{equation}
\begin{aligned}
\bar{r}[\ell] &\triangleq \sum_{k=0}^{K-1}w[k]\mathrm{e}^{\mathrm{j}2\uppi k\ell/K}, \\
\widetilde{r}[\ell] &\triangleq \sum_{k\in\mathcal{S}_{\mathrm{d}}}
p_{\mathrm{d},k}(|s_k|^2-1)\mathrm{e}^{\mathrm{j}2\uppi k\ell/K}.
\end{aligned}
\end{equation}
By Assumption~\ref{assume1}, $\mathbb{E}[\widetilde{r}[\ell]]=0$, so the cross term between $\bar r[\ell]$ and $\widetilde r[\ell]$ vanishes. Expanding $\mathbb{E}\!\left[|\widetilde{r}[\ell]|^2\right]$ yields
\begin{equation}
\begin{aligned}
\mathbb{E}\!\left[|\widetilde{r}[\ell]|^2\right]
&= \sum_{k\in\mathcal{S}_{\mathrm{d}}}\sum_{m\in\mathcal{S}_{\mathrm{d}}}
p_{\mathrm{d},k}p_{\mathrm{d},m}\mathrm{e}^{\mathrm{j}2\uppi(k-m)\ell/K} \\
&\times \mathbb{E}\!\left[(|s_k|^2-1)(|s_m|^2-1)\right] \\
&= (\mu_4-1)\sum_{k\in\mathcal{S}_{\mathrm{d}}}p_{\mathrm{d},k}^2.
\end{aligned}
\end{equation}
The MS P-ACF is then given by
\begin{equation}
\begin{aligned}
\mathbb{E}\!\left[|r[\ell]|^2\right]
&= \mathbb{E}\!\left[\left|\bar{r}[\ell]+\widetilde{r}[\ell]\right|^2\right] \\
&= |\bar{r}[\ell]|^2+\mathbb{E}\!\left[|\widetilde{r}[\ell]|^2\right] \\
&= \left|\sum_{k=0}^{K-1}w[k]\mathrm{e}^{\mathrm{j}2\uppi k\ell/K}\right|^2
+(\mu_4-1)\sum_{k\in\mathcal{S}_{\mathrm{d}}}p_{\mathrm{d},k}^2.
\end{aligned}
\end{equation}

\section{Proof of Corollary~\ref{coro1}}\label{appeb}
Using the Parseval relation for the DFT and
$|\Omega_{\mathrm g}|=K-1$, we obtain
\begin{equation}
\begin{aligned}
\mathrm{EISL}(\Omega_{\mathrm g})
&=
K\sum_{k=0}^{K-1}w^2[k]-P^2
+(K-1)(\mu_4-1)
\sum_{k\in\mathcal S_{\mathrm d}}p_{\mathrm d,k}^2\\
&=
K\sum_{k\in\mathcal S_{\mathrm p}}p_{\mathrm p,k}^2
+\left[(K-1)\mu_4+1\right]
\sum_{k\in\mathcal S_{\mathrm d}}p_{\mathrm d,k}^2
-P^2.
\end{aligned}
\end{equation}

\section{Proof of Theorem~\ref{theoremgeisl}}\label{appec}

For \eqref{eislopti}, define the total pilot and data powers as
\begin{equation}
P_{\mathrm{p}}=\sum_{k\in\mathcal{S}_{\mathrm{p}}}p_{\mathrm{p},k}, 
P_{\mathrm{d}}=\sum_{k\in\mathcal{S}_{\mathrm{d}}}p_{\mathrm{d},k}.
\end{equation}
For notational convenience, define
\begin{equation}
Q_{\mathrm{p}}=\sum_{k\in\mathcal{S}_{\mathrm{p}}}p_{\mathrm{p},k}^2, 
Q_{\mathrm{d}}=\sum_{k\in\mathcal{S}_{\mathrm{d}}}p_{\mathrm{d},k}^2.
\end{equation}
Then the objective function of \eqref{eislopti} can be written as
\begin{equation}
F(Q_{\mathrm{p}},Q_{\mathrm{d}})
=
\frac{KQ_{\mathrm{p}}+\beta Q_{\mathrm{d}}-P^2}{P^2+(\mu_4-1)Q_{\mathrm{d}}}.
\end{equation}
Its partial derivatives are
\begin{equation}
\begin{aligned}
\frac{\partial F}{\partial Q_{\mathrm p}}
&=
\frac{K}{P^2+(\mu_4-1)Q_{\mathrm d}}>0,\\
\frac{\partial F}{\partial Q_{\mathrm d}}
&=
\frac{K\mu_4P^2-K(\mu_4-1)Q_{\mathrm p}}
{\left[P^2+(\mu_4-1)Q_{\mathrm d}\right]^2}>0.
\end{aligned}
\end{equation}
Hence, $F$ is strictly increasing in both $Q_{\mathrm p}$ and $Q_{\mathrm d}$. By the Cauchy--Schwarz inequality,
\begin{equation}
Q_{\mathrm p}
\geq
\frac{P_{\mathrm p}^2}{|\mathcal S_{\mathrm p}|},
\qquad
Q_{\mathrm d}
\geq
\frac{P_{\mathrm d}^2}{|\mathcal S_{\mathrm d}|},
\end{equation}
where equality holds if and only if the powers are uniform within the
corresponding pilot and data sets. Thus, every optimum has the two-level
structure
\begin{equation}
p_{\mathrm p,k}=p_{\mathrm p},\quad k\in\mathcal S_{\mathrm p},
\qquad
p_{\mathrm d,k}=p_{\mathrm d},\quad k\in\mathcal S_{\mathrm d}.
\end{equation}
For any given $P_{\mathrm{p}}$ and $P_{\mathrm{d}}$, the objective function is minimized when $Q_{\mathrm{p}}$ and $Q_{\mathrm{d}}$ attain their lower bounds, namely, when power is uniformly allocated within the P-D subcarrier sets. Hence, every optimal solution to \eqref{eislopti} has a two-level structure
\begin{equation}
\begin{aligned}
p_{\mathrm{p},k}&=p_{\mathrm{p}},\; k\in\mathcal{S}_{\mathrm{p}},\\
p_{\mathrm{d},k}&=p_{\mathrm{d}},\; k\in\mathcal{S}_{\mathrm{d}}.
\end{aligned}
\end{equation}
Let
\begin{equation}
P_{\mathrm{d}}=tP, P_{\mathrm{p}}=(1-t)P, 0\leq t\leq1.
\end{equation}
Then
\begin{equation}
p_{\mathrm{d}}=\frac{tP}{|\mathcal{S}_{\mathrm{d}}|}, 
p_{\mathrm{p}}=\frac{(1-t)P}{|\mathcal{S}_{\mathrm{p}}|}.
\end{equation}
Substituting into the objective function reduces \eqref{eislopti} to
\begin{equation}
\min_{0\leq t\leq1}f(t)
=
\frac{\dfrac{K}{|\mathcal{S}_{\mathrm{p}}|}(1-t)^2+\dfrac{\beta}{|\mathcal{S}_{\mathrm{d}}|}t^2-1}
{1+\dfrac{\mu_4-1}{|\mathcal{S}_{\mathrm{d}}|}t^2}.
\end{equation}
Examining $f'(t)$ and $f''(t)$ shows that $f(t)$ first decreases and then increases over $[0,1]$. Setting $f'(t^\star)=0$ yields the unique global minimizer \eqref{tstar}. Substituting $t^\star$ yields the optimal P-D subcarrier powers $p_{\mathrm{p}}^\star$ and $p_{\mathrm{d}}^\star$.

Next, consider UPA over all subcarriers. In this case, $t=\frac{|\mathcal{S}_{\mathrm{d}}|}{K}$, and $f'\!\left(\frac{|\mathcal{S}_{\mathrm{d}}|}{K}\right)\geq0$, with equality if and only if $\mu_4=1$. This implies that $t^\star\leq\frac{|\mathcal{S}_{\mathrm{d}}|}{K}$. Expanding $p_{\mathrm{p}}^\star$ and $p_{\mathrm{d}}^\star$ gives
\begin{equation}
p_{\mathrm{d}}^\star=\frac{Pt^\star}{|\mathcal{S}_{\mathrm{d}}|}\leq\frac{P}{K}\leq\frac{(1-t^\star)P}{|\mathcal{S}_{\mathrm{p}}|}=p_{\mathrm{p}}^\star.
\end{equation}

\section{Proof of Proposition~\ref{prop4}}\label{apped}

Expanding the deterministic term of the ROI-EISL yields
\begin{align}
&\sum_{\ell\in\Omega_\xi}
\left|\sum_{k=0}^{K-1}w[k]\mathrm{e}^{\mathrm{j}2\uppi k\ell/K}\right|^2 \notag\\
&=\sum_{\ell\in\Omega_\xi}
\left(\sum_{m=0}^{K-1}w[m]\mathrm{e}^{\mathrm{j}2\uppi m\ell/K}\right)
\left(\sum_{n=0}^{K-1}w[n]\mathrm{e}^{-\mathrm{j}2\uppi n\ell/K}\right) \notag\\
&=\sum_{m=0}^{K-1}\sum_{n=0}^{K-1}
w[m]w[n]\sum_{\ell\in\Omega_\xi}\mathrm{e}^{\mathrm{j}2\uppi(m-n)\ell/K}.
\end{align}
By substituting $[\boldsymbol{C}_{\Omega_\xi}]_{m,n}$, we obtain
\begin{equation}
\sum_{\ell\in\Omega_\xi}
\left|\sum_{k=0}^{K-1}w[k]\mathrm{e}^{\mathrm{j}2\uppi k\ell/K}\right|^2
=\boldsymbol{w}^{\mathrm{T}}\boldsymbol{C}_{\Omega_\xi}\boldsymbol{w}.
\end{equation}
Expanding the random term gives
\begin{equation}
|\Omega_\xi|(\mu_4-1)\sum_{k\in\mathcal{S}_{\mathrm{d}}}p_{\mathrm{d},k}^2
=|\Omega_\xi|(\mu_4-1)\boldsymbol{w}^{\mathrm{T}}
\boldsymbol{D}_{\mathrm{d}}(\mathcal{S}_{\mathrm{p}})\boldsymbol{w}.
\end{equation}
Substituting these results into the definition of the normalized ROI-EISL gives
\begin{equation}
\mathrm{EISL}^{\mathrm{N}}(\Omega_\xi)
=
\frac{
\boldsymbol{w}^{\mathrm{T}}\!\left[
\boldsymbol{C}_{\Omega_\xi}
+|\Omega_\xi|(\mu_4-1)\boldsymbol{D}_{\mathrm{d}}(\mathcal{S}_{\mathrm{p}})
\right]\boldsymbol{w}
}{
P^2+(\mu_4-1)\boldsymbol{w}^{\mathrm{T}}
\boldsymbol{D}_{\mathrm{d}}(\mathcal{S}_{\mathrm{p}})\boldsymbol{w}
}.
\end{equation}

\section{Proof of Proposition~\ref{prop5}}\label{appee}

Since $\boldsymbol{C}_{\Omega_\xi}\succeq\mathbf{0}$ and
$p_{\mathrm d,k}\geq p_{\mathrm{min}}$,
\begin{equation}
\begin{aligned}
\mathrm{EISL}^{\mathrm N}(\Omega_\xi)
&\geq
\frac{
|\Omega_\xi|(\mu_4-1)q
}{
P^2+(\mu_4-1)q
}\\
&\geq
\frac{
|\Omega_\xi|(\mu_4-1)
|\mathcal S_{\mathrm d}|p_{\mathrm{min}}^2
}{
P^2+(\mu_4-1)
|\mathcal S_{\mathrm d}|p_{\mathrm{min}}^2
},
\end{aligned}
\end{equation}
where
$q=\boldsymbol{w}^{\mathrm T}
\boldsymbol{D}_{\mathrm d}(\mathcal S_{\mathrm p})
\boldsymbol{w}
=\sum_{k\in\mathcal S_{\mathrm d}}p_{\mathrm d,k}^2
\geq|\mathcal S_{\mathrm d}|p_{\mathrm{min}}^2$.
The second inequality follows because
$q/[P^2+(\mu_4-1)q]$ is nondecreasing for $q\geq0$.

\section{Proof of Theorem~\ref{theorem2}}\label{appef}

For the strictly uniform pilot placement and $\ell\in\Omega_\xi$,
\begin{align}\label{pilotf}
\sum_{k\in\mathcal{S}_{\mathrm{p}}^{\star}}
\mathrm{e}^{\mathrm{j}2\uppi k\ell/K}
&=
\mathrm{e}^{\mathrm{j}2\uppi k_0\ell/K}
\sum_{m=0}^{|\mathcal{S}_{\mathrm{p}}|-1}
\mathrm{e}^{\mathrm{j}2\uppi m\ell/|\mathcal{S}_{\mathrm{p}}|} \notag\\
&=0,
\end{align}
where the second equality follows from the finite geometric series. We further consider
\begin{align}
&\sum_{k=0}^{K-1}w^{\star}[k]\mathrm{e}^{\mathrm{j}2\uppi k\ell/K}\notag \\ 
&=
p_{\mathrm{min}}\sum_{k=0}^{K-1}\mathrm{e}^{\mathrm{j}2\uppi k\ell/K}
+\left(p_{\mathrm{p}}^{\star}-p_{\mathrm{min}}\right)
\sum_{k\in\mathcal{S}_{\mathrm{p}}^{\star}}
\mathrm{e}^{\mathrm{j}2\uppi k\ell/K} \notag\\
&=0.
\end{align}
The second equality follows from the orthogonality of DFT basis functions and \eqref{pilotf}. Therefore, $(\boldsymbol{w}^{\star})^{\mathrm{T}}\boldsymbol{C}_{\Omega_\xi}\boldsymbol{w}^{\star}=0$. Moreover, all data subcarriers are assigned their minimum power, yielding $(\boldsymbol{w}^{\star})^{\mathrm{T}}
\boldsymbol{D}_{\mathrm{d}}(\mathcal{S}_{\mathrm{p}}^{\star})
\boldsymbol{w}^{\star}
=
|\mathcal{S}_{\mathrm{d}}|p_{\mathrm{min}}^2$. Substituting these two results into the normalized ROI-EISL yields
\begin{align}
\mathrm{EISL}^{\mathrm{N}}(\Omega_\xi)
&=
\frac{
(\boldsymbol{w}^{\star})^{\mathrm{T}}\boldsymbol{C}_{\Omega_\xi}\boldsymbol{w}^{\star}
+|\Omega_\xi|(\mu_4-1)(\boldsymbol{w}^{\star})^{\mathrm{T}}
\boldsymbol{D}_{\mathrm{d}}(\mathcal{S}_{\mathrm{p}}^{\star})
\boldsymbol{w}^{\star}
}{
P^2+(\mu_4-1)(\boldsymbol{w}^{\star})^{\mathrm{T}}
\boldsymbol{D}_{\mathrm{d}}(\mathcal{S}_{\mathrm{p}}^{\star})
\boldsymbol{w}^{\star}
} \notag\\
&=
\frac{
|\Omega_\xi|(\mu_4-1)|\mathcal{S}_{\mathrm{d}}|p_{\mathrm{min}}^2
}{
P^2+(\mu_4-1)|\mathcal{S}_{\mathrm{d}}|p_{\mathrm{min}}^2
}.
\end{align}

\section{Proof of Proposition~\ref{prop6}}\label{appeg}

For fixed $\mathcal S_{\mathrm p}$, let
$F(\boldsymbol{w})=\mathrm{EISL}^{\mathrm N}(\Omega_\xi)$.
For any $\eta\in[0,|\Omega_\xi|]$,
$F(\boldsymbol{w})\leq\eta$ is equivalent to
\begin{equation}
\boldsymbol{w}^{\mathrm T}
\left[
\boldsymbol{C}_{\Omega_\xi}
+(\mu_4-1)(|\Omega_\xi|-\eta)
\boldsymbol{D}_{\mathrm d}(\mathcal S_{\mathrm p})
\right]
\boldsymbol{w}
\leq
\eta P^2,
\end{equation}
i.e.,
$\boldsymbol{w}^{\mathrm T}\boldsymbol{A}_{\eta}\boldsymbol{w}
\leq\eta P^2$.
For fixed $\mathcal S_{\mathrm p}$,
$\boldsymbol{A}_{\eta}\succeq\mathbf 0$ over this interval, so the
corresponding feasibility problem is a convex QP. Furthermore, feasibility
is monotonic in $\eta$: if $\eta$ is feasible, every $\eta'\geq\eta$ is
also feasible. Hence,
\begin{equation}
F^\star
=
\min
\left\{
\eta\in[0,|\Omega_\xi|]:
q(\eta)\leq\eta P^2
\right\},
\end{equation}
which can be obtained by bisection.

\bibliographystyle{IEEEtran}
\bibliography{books1}

\end{document}